\documentclass[conference,letterpaper]{IEEEtran2026}
\usepackage[left=0.74in,right=0.64in,top=0.69in,bottom=0.99in]{geometry}

\IEEEoverridecommandlockouts
\newcommand{\Latexfilepath}{.}
\usepackage{\Latexfilepath/Bermudez} 

\begin{document}

\title{Equivalence Between Nested Gibbs Measures and \\ Log-Linear Combinations of Gibbs Measures}

 \author{
   \IEEEauthorblockN{
   Yaiza Bermudez\IEEEauthorrefmark{1}, 
   Samir M. Perlaza\IEEEauthorrefmark{1}\IEEEauthorrefmark{2}\IEEEauthorrefmark{4}, and
   I\~{n}aki Esnaola\IEEEauthorrefmark{3}\IEEEauthorrefmark{4}\\
   Emails: name.lastname@inria.fr and esnaola@sheffield.ac.uk
                     }
   \IEEEauthorblockA{\IEEEauthorrefmark{1}%
                     Centre Inria d'Université Côte d'Azur, INRIA,
                     Sophia Antipolis, France.}
   \IEEEauthorblockA{\IEEEauthorrefmark{2}%
                     Laboratoire GAATI, Universit\'{e} de la Polyn\'{e}sie fran\c{c}aise, Fa`a`\={a}, French Polynesia.}
   \IEEEauthorblockA{\IEEEauthorrefmark{3}%
                     School of Electrical and Electronic Engineering, University of Sheffield, Sheffield,
United Kingdom.}   
   \IEEEauthorblockA{\IEEEauthorrefmark{4}%
		     Department of Electrical and Computer Engineering, Princeton University, Princeton, NJ 08544, USA.
\thanks{This work is supported in part by the European Commission through the H2020-MSCA-RISE-2019 project 872172; the French National Agency for Research (ANR)  through the Project ANR-21-CE25-0013 and the project ANR-22-PEFT-0010 of the France 2030 program PEPR Réseaux du Futur; and in part by the Agence de l'innovation de défense (AID) through the project UK-FR 2024352. }                     } 
  }
 \maketitle

\begin{abstract} 
In this paper, three operations on Gibbs probability measures are studied. 
The first operation, which takes as argument one Gibbs probability measure and is often referred to as renormalization, consists of generating a new Gibbs measure by normalizing a power of the density of the given measure. Such a normalization has a twofold effect: first, it changes the regularization factor; and second, it concentrates the support within a subset of the original support. Interestingly, it is shown that these effects can be independently controlled by different parameters. 
The second operation consists of a normalized log-linear combination of the densities of Gibbs probability measures.
The third operation, which takes as argument two Gibbs probability measures, consists of changing the reference measure of the latter with the former. Hence, the former is said to be ``nested'' within the latter, yielding a new Gibbs probability measure. 
The resulting probability measures, from both the second and third operations, which are also Gibbs probability measures, are shown, respectively, to be the solutions to optimization problems of the expectations of linear combinations of the objective functions of the given measures, subject to a relative entropy regularization. Such optimization problems differ exclusively in the coefficients of the linear combinations. This observation leads to the conclusion that there exists a set of parameters for which nesting one Gibbs probability measure into another has the same effect as log-linearly combining them. 
These operations are shown to have relevant applications in statistical learning. As an example, a one-shot federated learning system in which clients send their locally trained Gibbs algorithms to the server to be log-linearly combined by the server is shown to achieve the same performance as a Gibbs algorithm trained upon the aggregation of all local training datasets.
\end{abstract}
\vspace{-0.05in}
\enlargethispage{-0.05in}
\section{Introduction}
\vspace{-0.08in}
Gibbs probability measures have gained significant attention in statistical learning theory and information theory, in part due to their numerous properties  \cite{xu2017information, russo2016controlling, medina2022robustness, Bu2024Towards, bu2023generalization, perlaza2023validation, catoni2007pac, ray2023asymptotics, perlaza2024empirical, rodriguezgalvez2024information} and the fact that they represent the long-run behavior of stochastic-gradient-based learning algorithms with time-invariant learning rates~\cite{Azizian2024What}.
%
Gibbs measures form a benchmark to which algorithms can be compared for assessing their generalization error
\cite{aminian2021exact, aminian2024information, perlaza2024generalization, zou2024Generalization} and have also been shown to represent the worst-case data-generating probability distributions for a fixed model in supervised learning~\cite{zouJSAIT2024}.
%

The main contributions of this work concern three operations on Gibbs probability measures.
The first operation consists of constructing a Gibbs probability measure by normalizing (to one) a power of the density of a given Gibbs probability measure. This operation is classical in statistical physics and information theory, and is not limited to Gibbs measures. See for instance \cite{bercher2012simple, wilson1983renormalization, asadi2025hierarchical, chhabra1989direct, beck1993thermodynamics, abe2003geometry, ohara2010dually, bercher2009source, tsallis2009introduction, abe2005necessity} and references therein. While some authors refer to such operation as ``renormalization'', others refer to the resulting probability distribution as an ``escort distribution''.
In this work, renormalization is shown to lead to two independent effects on the original Gibbs measure: first, trimming the regularization factor; and second, concentrating the support on a subset of its original support (Theorem~\ref{TheoremApril16at13h13in2026Antibes}). Interestingly, both effects are shown to be independently controlled by different parameters of this operation. This new result complements existing evidence on the benefits of trimming the regularization factor in~\cite{Bu2024Towards, aminian2021exact, perlaza2023validation}, as well as on the benefits of strategically choosing the reference measure as shown in \cite{bermudez2026machine, perlaza2024empirical,bermudez2026decentralized}.
The second operation is the normalized log-linear combination of powers of densities of Gibbs probability measures. This generalization of renormalization has been studied beyond Gibbs probability measures in \cite{asadi2025hierarchical, lalitha2018social, bercher2012simple} and references therein. 
In this work, it is shown that the normalized log-linear combination of two powers of densities of Gibbs measures is itself a Gibbs measure, and that it is the solution to a minimization of the expectation of a linear combination of the objective functions of the original Gibbs measures subject to a regularization by a relative entropy (Theorem~\ref{TheoremApril16at15h19in2026Antibes}).
\enlargethispage{-0.10in}

Given two Gibbs probability measures, the third operation consists of constructing a new Gibbs probability measure by using the first measure as the reference measure of the second. The first measure is then said to be nested within the second, yielding a \emph{nested Gibbs probability measure}. Nested Gibbs measures form a special class of Gibbs measures with important applications in statistical learning theory, including decentralized learning and machine unlearning; see, for instance,~\cite{bermudez2026machine} and~\cite{bermudez2026decentralized}.
In this work, nested Gibbs measures are shown to be solutions to a minimization of the expectation of a linear combination of the two original objective functions subject to a regularization by a relative entropy with respect to a given reference measure (Theorem~\ref{TheoremApril16at15h22in2026Antibes}).
\enlargethispage{-0.05in}
Hence, nesting Gibbs probability measures or log-linearly combining them leads to new Gibbs probability measures that are solutions to optimization problems whose difference is only on the coefficients of the linear combination of such objective functions. Using this observation, the central result is a formal proof that there exists a choice of parameters such that both operations yield the same Gibbs probability measure (Corollary~\ref{CorollaryApril16at15h20in2026Antibes}).
Finally, this equivalence between such operations is used to design a one-shot federated learning system \cite{mcmahan2017communication} that achieves, at all clients, the same performance as a centralized system that trains a Gibbs algorithm upon the aggregation of the training datasets of all clients (Corollary~\ref{CorApril23at18h55in2026Antibes}). This result is reminiscent of the decentralized learning construction in~\cite{bermudez2026decentralized}, which also achieves the same performance as a centralized system. 
While the construction in~\cite{bermudez2026decentralized} is based on sequentially nesting the Gibbs algorithms of all clients, the construction in this paper is based on simultaneously log-linearly combining the clients' Gibbs algorithms at a central server.
\vspace{-1.5ex}
\section{Preliminaries}
This section describes the notation used in this work and formally introduces conditional Gibbs probability measures.
\enlargethispage{-0.10in}
\subsection{Notation}
Given a set $\set{X} \subseteq \reals^d$, for some $d \in \ints$, the Borel sigma-field defined on $\set{X}$ is denoted by $\mathscr{B}\autoparent{\set{X}}$. 
The set of probability measures on the measurable space $\left(\set{X},\mathscr{B}(\set{X})\right)$
is denoted by $\triangle(\set{X})$. 
Given a product measurable space $\left(\set{X}\times\set{Y},\mathscr{B}(\set{X}\times\set{Y},)\right)$, the set of  probability measures on $\set{Y}$ conditioned on $\set{X}$ is denoted by $\simplex{\set{
Y} | \set{X}}$.
Given a $\sigma$-finite measure $Q$ on  $\left(\set{X},\mathscr{B}(\set{X})\right)$, the set of probability measures in $\simplex{\set{X}}$ that are absolutely continuous with respect to $Q$ is denoted by $\simplexabs{Q}{\set{X}}$.  
Given two measures $P$ and $Q$ on the same measurable space, the notation $\abscontPQ$ stands for ``the measure $P$ is absolutely continuous with respect to $Q$''. The Radon-Nikodym derivative of $P$ with respect to $Q$ is denoted by $\RND{P}{Q}$. The relative entropy of $P$ with respect to $Q$ is denoted by $\KL{P}{Q}$.
 
\subsection{Gibbs Conditional Probability Measures}\label{SecGibbsConditionalProbMeasures}

Consider two sets $\set{X} \subset \reals^{d_1}$ and $\set{Y} \subset \reals^{d_2}$, for some given integers $d_1$ and $d_2$. Let $\left(\set{X},\mathscr{B}\autoparent{\set{X}}\right)$ and $\left(\set{Y},\mathscr{B}\autoparent{\set{Y}}\right)$  be two Borel measurable spaces. Consider also the product measurable space $(\set{X} \times \set{Y}, \mathscr{B}(\set{X} \times  \set{Y}))$. 
Using this notation, Gibbs conditional probability measures in $\simplex{\set{Y} | \set{X}}$ are parametrized by a Borel measurable function
$h: \set{X} \times \set{Y} \to \mathbb{R}$; a $\sigma$-finite measure $Q$ on the measurable space~$(\set{Y}, \mathscr{B}(\set{Y}))$; and a real $\lambda \in \reals\setminus\lbrace 0 \rbrace$.
The function $h$ is referred to as the objective function; the parameter $\lambda$ is referred to as the regularization factor; and the measure $Q$ is often referred to as the reference measure. This nomenclature will become more evident as the connections of Gibbs measures to specific optimization problems are unveiled.
The definition of Gibbs conditional probability measures is presented hereunder.
\begin{definition}[Gibbs Conditional Probability Measures]
\label{DefGeneralGibbs}
Given a Borel measurable function~$h : \set{X}\times\set{Y}\to\reals$; a~$\sigma$-finite measure~$Q$ on the measurable space~$(\set{Y}, \mathscr{B}(\set{Y}))$; and
a real~$\lambda \in \reals\setminus\lbrace 0 \rbrace$, the probability measure~$P^{(h, Q, \lambda)}_{Y \mid X}\in \triangle\left(\set{Y} \mid \set{X} \right)$ is said to be an~$(h, Q, \lambda)$-Gibbs conditional probability measure if
\vspace{-3ex}
\begin{IEEEeqnarray}{rCl}
\label{EqGeneralGibbsIntegrability}
\mbox{$\forall x \in \set{X}$, } 0 < \int \exp\left( -\frac{1}{\lambda}h\autoparent{x,y}\right)\mathrm{d}Q\left( y \right) < + \infty;
\end{IEEEeqnarray}
and for all~$\left(x, y \right) \in \set{X} \times \supp Q$,
\vspace{-1ex}
\begin{IEEEeqnarray}{rcl}
\label{EqGeneralGibbsRN}
\frac{\mathrm{d}P_{Y\mid X = x}^{(h, Q, \lambda)}}{\mathrm{d}Q } \autoparent{y} & = &\frac{\exp\autoparent{- \frac{1}{\lambda} h\autoparent{x,y}}}{ \int \exp\left( -\frac{1}{\lambda}h\autoparent{x,y}\right)\mathrm{d}Q\left( y \right)}.
\end{IEEEeqnarray}
\end{definition}
In Definition~\ref{DefGeneralGibbs}, while $P_{Y|X}^{(h,Q,\lambda)}$ is an $(h, Q, \lambda)$-Gibbs conditional probability measure,  the measure $P_{Y|X=x}^{(h,Q,\lambda)}$, obtained by conditioning it upon a given vector $x\in\set{X}$, is referred to as an $(h, Q, \lambda)$-Gibbs probability measure.

\section{Main Results}

\subsection{Renormalization of Gibbs Probability Measures}
\label{SecRenormNested} 

The normalization of a power of a probability measure, which is often referred to as \emph{renormalization}, is defined as follows.
\begin{definition}[Renormalization]
\label{DefRenormalization}
Consider $\alpha \in \reals\setminus\lbrace 0 \rbrace$, two $\sigma$-finite measures $Q_1$ and $Q_2$, and a probability measure $P$, all on the same measurable space with $P \ll Q_1$ and $Q_2 \ll Q_1$. The $\autoparent{\alpha,Q_1, Q_2}$-renormalization of $P$ is a probability measure,  denoted by $R$, such that for all $y \in \supp Q_2$,
\begin{IEEEeqnarray}{rCl}
\label{EqRenormalization}
\RND{R}{Q_2}\autoparent{y}
&=&
\frac{\autoparent{\RND{P}{Q_1}\autoparent{y}}^{\alpha}}
{\displaystyle\int \autoparent{\RND{P}{Q_1}\autoparent{\tilde{y}}}^{\alpha}\mathrm{d}Q_2\autoparent{\tilde{y}}},
\end{IEEEeqnarray}
provided that 
$0
<
\int \autoparent{\RND{P}{Q_1}\autoparent{y}}^{\alpha}\mathrm{d}Q_2\autoparent{y}
< + \infty.
$
\end{definition}
\enlargethispage{-0.10in}
The renormalization of a Gibbs probability measure is also a Gibbs probability measure, as shown hereunder.
More specifically, given the $(h,Q,\lambda)$-Gibbs probability measure $P^{(h,Q,\lambda)}_{Y\mid X=x}$ in~\eqref{EqGeneralGibbsRN}, denote by $R_1$, its $(\alpha, Q, Q_1)$-renormalization, which satisfies, for all $y \in \supp Q_1$,
\begin{IEEEeqnarray}{rCl}
\label{EqApril15at17h43in2026Sophia}
\RND{R_1}{Q_1}\autoparent{y} = \frac{\autoparent{\RND{P^{(h,Q,\lambda)}_{Y\mid X=x}}{Q} \autoparent{y}}^{\alpha}}{\displaystyle\int \autoparent{\RND{P^{(h,Q,\lambda)}_{Y\mid X=x}}{Q}\autoparent{\hat{y}}}^{\alpha} \mathrm{d}Q_1\autoparent{\hat{y}}}.
\end{IEEEeqnarray}
The following theorem shows that $R_1$ is a Gibbs probability measure.
\begin{theorem}
\label{TheoremApril16at13h13in2026Antibes}
The $(\alpha, Q, Q_1)$-renormalization of the Gibbs probability measure $P^{(h,Q,\lambda)}_{Y\mid X=x}$ in \eqref{EqGeneralGibbsRN}, namely the measure $R_{1}$ in \eqref{EqApril15at17h43in2026Sophia}, satisfies for all $y \in \supp Q_1$, 
\begin{IEEEeqnarray}{rCl}
\label{EqApril16at13h03in2026Antibes}
\RND{R_1}{Q_1}\autoparent{y} &=& \frac{\exp\autoparent{-~\frac{\alpha}{\lambda} h(x,y)}}{\int \exp\autoparent{-~\frac{\alpha}{\lambda} h(x,\hat{y})} \mathrm{d}Q_1(\hat{y})} =  \RND{P^{(h,Q_1,\frac{\lambda}{\alpha})}_{Y\mid X=x}}{Q_1}\autoparent{y}. \quad
\end{IEEEeqnarray}
\end{theorem}
\begin{IEEEproof}
The proof is presented in \cite[Appendix~C]{InriaRR9613}.
\end{IEEEproof}

Theorem~\ref{TheoremApril16at13h13in2026Antibes} shows that the $(\alpha, Q, Q_1)$-renormalization of the $(h,Q,\lambda)$-Gibbs probability measure $P^{(h,Q,\lambda)}_{Y\mid X=x}$ is identical to the~$(h,Q_1,\frac{\lambda}{\alpha})$-Gibbs probability measure $P^{(h,Q_1,\frac{\lambda}{\alpha})}_{Y\mid X=x}$ in~\eqref{EqApril16at13h03in2026Antibes}. More specifically,
%
$\KL{R_1}{P^{(h,Q_1,\frac{\lambda}{\alpha})}_{Y\mid X=x}}  =  0$.
%
Hence, such a renormalization has a twofold effect on the Gibbs measure $P^{(h,Q,\lambda)}_{Y\mid X=x}$ in~\eqref{EqGeneralGibbsRN}: (i) the regularization factor is changed from $\lambda$ to $\lambda/\alpha$; and (ii) the reference measure is changed from~$Q$ to~$Q_1$.  The second effect might imply a concentration of the measure given that $Q_1 \ll Q$. More specifically if $\supp Q_1 \subset \supp Q$, then, $\supp P^{(h,Q_1,\frac{\lambda}{\alpha})}_{Y\mid X=x}  = \supp Q_1 \subset \supp  Q =  \supp P^{(h,Q,\lambda)}_{Y\mid X=x}$. 
Interestingly, both effects are not necessarily simultaneously observed, which follows from the fact that $\alpha$ and $Q_1$ can be chosen independently. For instance, if $Q_1$ is chosen to be identical to $Q$, then, only the change of the regularization factor is observed. Alternatively, if $\alpha = 1$, only the change of reference measure is observed.
This observation reveals a structural property of the following optimization problems
\begin{IEEEeqnarray}{rCl}
\label{EqMarch5at14h56in2026SophiaA}
\left\lbrace
\begin{array}{l}
\displaystyle\min_{P\in\simplexabs{Q}{\set{Y}}}  \int h\autoparent{x,y}\,\mathrm{d}P(y) + \lambda \KL{P}{Q},  \mbox{ $\lambda > 0$} \\
\displaystyle\max_{P\in\simplexabs{Q}{\set{Y}}}  \int h\autoparent{x,y}\,\mathrm{d}P(y)+ \lambda \KL{P}{Q},  \mbox{ $\lambda < 0$},
\end{array}
\right.
\end{IEEEeqnarray}
whose unique solutions, if they exist, are both identical to the Gibbs probability measure $P_{Y\mid X = x}^{(h, Q, \lambda)}$ in~\eqref{EqGeneralGibbsRN}. See \cite[Lemma~1]{PerlazaEntropy2025}.
\enlargethispage{-0.10in}
More specifically, if the solution to one of the optimization problems in \eqref{EqMarch5at14h56in2026SophiaA} is known, an $(\alpha, Q, Q_1)$-renormalization of such a solution leads to the solution, if it exists, to one of the following optimization problems:
\begin{IEEEeqnarray}{rCl}
\label{EqApril19at22h17in2026Nice}
\left\lbrace
\begin{array}{l}
\displaystyle\min_{P\in\simplexabs{Q_1}{\set{Y}}}  \int h\autoparent{x,y}\,\mathrm{d}P(y) \;+\;\frac{\lambda}{\alpha}\,\KL{P}{Q_1},  \mbox{ $\frac{\lambda}{\alpha} > 0$} \\
\displaystyle\max_{P\in\simplexabs{Q_1}{\set{Y}}}  \int h\autoparent{x,y}\,\mathrm{d}P(y) \;+\;\frac{\lambda}{\alpha}\,\KL{P}{Q_1},  \mbox{ $\frac{\lambda}{\alpha} < 0$}.
\end{array}
\right. 
\end{IEEEeqnarray}

\subsection{Normalized Log-Linear Combinations of Gibbs Measures}
\label{SecLogLinearMain}

The normalization of the product of powers of probability measures is formalized as follows. 
\begin{definition}[Normalized Log-Linear Combination of Measures]
\label{DefApril14at23h55in2026Antibes}
Consider $\alpha_1 \in \reals\setminus\lbrace 0 \rbrace$, $\alpha_2 \in \reals\setminus\lbrace 0 \rbrace$, two probability measures $P_1$ and $P_2$; and three $\sigma$-finite measures $Q$, $Q_1$ and $Q_2$, all on the same measurable space, such that $P_1 \ll Q_1$, $Q\ll Q_1$, $P_2 \ll Q_2$ and $Q\ll Q_2$. The $(\alpha_1, \alpha_2, Q_1, Q_2, Q)$-normalized log-linear combination of $P_1$ and $P_2$ is a probability measure, denoted by $S$, such that for all $y \in \supp Q$,
\begin{IEEEeqnarray}{rCl}
\label{EqApril15at18h35in2026Sophia}
\RND{S}{Q}\autoparent{y}
&=&
\frac{\autoparent{\RND{P_1}{Q_1}\autoparent{y}}^{\alpha_1} \autoparent{\RND{P_2}{Q_2}\autoparent{y}}^{\alpha_2}}
{\displaystyle\int \autoparent{\RND{P_1}{Q_1}\autoparent{\hat{y}}}^{\alpha_1} \autoparent{\RND{P_2}{Q_2}\autoparent{\hat{y}}}^{\alpha_2}\mathrm{d}Q\autoparent{\hat{y}}},
\end{IEEEeqnarray}
provided that 
$0
<
{\displaystyle\int \autoparent{\RND{P_1}{Q_1}\autoparent{\hat{y}}}^{\alpha_1} \autoparent{\RND{P_2}{Q_2}\autoparent{\hat{y}}}^{\alpha_2}\mathrm{d}Q\autoparent{\hat{y}}} < + \infty.
$
\end{definition}
\enlargethispage{-0.0in}
Definition~\ref{DefApril14at23h55in2026Antibes}, which is restricted to the normalized log-linear combination of only two probability measures, can be generalized to normalized log-linear combinations of multiple probability measures. Nonetheless,5 for pedagogical purposes, this work focuses exclusively on the case of two Gibbs probability measures.
Consider two Borel measurable functions $
h_1: \set{X}\times\set{Y}\to\reals $ and $
h_2: \set{X}\times\set{Y}\to\reals$.
Given $x_1\in \set{X}$, consider an $(h_1,Q_1,\lambda_1)$-Gibbs probability measure, denoted by
\vspace{-1ex}
\begin{IEEEeqnarray}{rCl}
\label{EqApril15at18h05in2026Sophia}
P^{(h_1,Q_1,\lambda_1)}_{Y_1\mid X_1 = x_1} \in \simplex{\set{Y}}.
\end{IEEEeqnarray}
Given $x_2 \in \set{X}$, consider an $(h_2,Q_2,\lambda_2)$-Gibbs probability measure, denoted by
\vspace{-2ex}
\begin{IEEEeqnarray}{rCl}
\label{EqApril22at21h10in2026Nice}
P^{(h_2,Q_2,\lambda_2)}_{Y_2\mid X_2 = x_2} \in \simplex{\set{Y}}.
\end{IEEEeqnarray}
The $(\alpha_1, \alpha_2, Q_1, Q_2, Q)$-normalized log-linear combination of the Gibbs probability measures $P^{(h_1,Q_1,\lambda_1)}_{Y_1\mid X_1 = x_1}$ in \eqref{EqApril15at18h05in2026Sophia} and $P^{(h_2,Q_2,\lambda_2)}_{Y_2\mid X_2 = x_2}$ in~\eqref{EqApril22at21h10in2026Nice}, denoted by $S_1 \in \simplex{\set{Y}}$, satisfies for all $y \in \supp Q$
\vspace{-2.5ex}
\begin{IEEEeqnarray}{rCl}
\label{EqApril16at10h36in2026Antibes}
\RND{S_1}{Q}\hspace{-0.6ex}\autoparent{y}
&=&
\frac{\autoparent{\RND{P_{Y_1\mid X_1 = x_1}^{(h_1, Q_1, \lambda_1)}}{Q_1}\autoparent{y}}^{\alpha_1} \autoparent{\RND{P_{Y_2\mid X_2 = x_2}^{(h_2, Q_2, \lambda_2)}}{Q_2}\autoparent{y}}^{\alpha_2}}
{\displaystyle\int \autoparent{\RND{P_{Y_1\mid X_1 = x_1}^{(h_1, Q_1, \lambda_1)}}{Q_1}\autoparent{\tilde{y}}}^{\alpha_1} \autoparent{\RND{P_{Y_2\mid X_2 = x_2}^{(h_2, Q_2, \lambda_2)}}{Q_2}\autoparent{\tilde{y}}}^{\alpha_2} \mathrm{d}Q\autoparent{\tilde{y}}}.\middlesqueezeequ \spnum
\end{IEEEeqnarray}
The following theorem introduces an explicit expression for the Radon-Nikodym derivative $\RND{S_1}{Q}$ in \eqref{EqApril16at10h36in2026Antibes}. For doing so, consider the following function
\vspace{-1ex}
\begin{IEEEeqnarray}{rCl}
\label{EqEffectivehApril14}
\hat{h}_{\beta_1,\beta_2}
&:&
\longfunction{ \set{X} \times \set{X} \times \set{Y}}{\reals}{\autoparent{x_1,x_2,y}}
{\beta_1 h_1\autoparent{x_1,y}
+
\beta_2 h_2\autoparent{x_2,y},}
\end{IEEEeqnarray}
where $(\beta_1, \beta_2)$ is a pair of given parameters.
\begin{theorem}
\label{TheoremApril16at15h19in2026Antibes}
The $(\alpha_1, \alpha_2, Q_1, Q_2, Q)$-normalized log-linear combination of 
$P^{(h_1,Q_1,\lambda_1)}_{Y_1\mid X_1 = x_1}$ in \eqref{EqApril15at18h05in2026Sophia} and $P^{(h_2,Q_2,\lambda_2)}_{Y_2\mid X_2 = x_2}$ in~\eqref{EqApril22at21h10in2026Nice}, namely the probability measure $S_1$ in \eqref{EqApril16at10h36in2026Antibes},  satisfies for all $y \in \supp Q$,
\vspace{-2ex}
\begin{IEEEeqnarray}{rCl}
\label{EqApril15at22h01in2026Antibes}
\RND{S_1}{Q}\autoparent{y}
&=&
\frac{
\exp\autoparent{
-\frac{\alpha_1}{\lambda_1}h_1\autoparent{x_1,y}
-\frac{\alpha_2}{\lambda_2}h_2\autoparent{x_2,y}}}
{\int
\exp\autoparent{
-\frac{\alpha_1}{\lambda_1}h_1\autoparent{x_1,\hat{y}}
-\frac{\alpha_2}{\lambda_2}h_2\autoparent{x_2,\hat{y}}}
\mathrm{d}Q\autoparent{\hat{y}}} \spnum\\
& = & 
\RND{P_{Y\mid (X_1,X_2)=(x_1,x_2)}^{\autoparent{\hat{h}_{\frac{\alpha_1}{\lambda_1},\frac{\alpha_2}{\lambda_2}},Q,1}}}{Q} (y),
\vspace{-1ex}
\end{IEEEeqnarray}
where $P_{Y\mid (X_1,X_2)}^{\autoparent{\hat{h}_{\frac{\alpha_1}{\lambda_1},\frac{\alpha_2}{\lambda_2}},Q,1}}$ is an $(\hat{h}_{\frac{\alpha_1}{\lambda_1},\frac{\alpha_2}{\lambda_2}},Q_1,1)$-Gibbs conditional probability measure in $\simplex{\set{Y} | \set{X} \times \set{X}}$; and the function $\hat{h}_{\frac{\alpha_1}{\lambda_1},\frac{\alpha_2}{\lambda_2}}$ is defined in \eqref{EqEffectivehApril14}.
\end{theorem}
\begin{IEEEproof}
The proof is presented in \cite[Appendix F]{InriaRR9613}.
\end{IEEEproof}
\enlargethispage{-0.10in}
Theorem~\ref{TheoremApril16at15h19in2026Antibes} establishes that the $(\alpha_1, \alpha_2, Q_1, Q_2, Q)$-normalized log-linear combination of 
$P^{(h_1,Q_1,\lambda_1)}_{Y_1\mid X_1 = x_1}$ in \eqref{EqApril15at18h05in2026Sophia} and $P^{(h_2,Q_2,\lambda_2)}_{Y_2\mid X_2 = x_2}$ in~\eqref{EqApril22at21h10in2026Nice} is identical to an $\autoparent{\hat{h}_{\frac{\alpha_1}{\lambda_1}, \frac{\alpha_2}{\lambda_2}},Q,1}$-Gibbs probability measure $P_{Y\mid (X_1,X_2)=(x_1,x_2)}^{\autoparent{\hat{h}_{\frac{\alpha_1}{\lambda_1}, \frac{\alpha_2}{\lambda_2}},Q,1}}$. More formally, 
\begin{IEEEeqnarray}{rCl}
\KL{S_1}{P_{Y\mid (X_1,X_2)=(x_1,x_2)}^{\autoparent{\hat{h}_{\frac{\alpha_1}{\lambda_1}, \frac{\alpha_2}{\lambda_2}},Q,1}}} = 0.
\end{IEEEeqnarray}
From this perspective, \cite[Lemma~1]{PerlazaEntropy2025} leads to the  conclusion that the $(\alpha_1, \alpha_2, Q_1, Q_2, Q)$-normalized log-linear combination of 
$P^{(h_1,Q_1,\lambda_1)}_{Y_1\mid X_1 = x_1}$ and $P^{(h_2,Q_2,\lambda_2)}_{Y_2\mid X_2 = x_2}$ is the unique solution, if it exists, to the minimization problem:
%
\begin{IEEEeqnarray}{rCl}
\label{EqApril22at22h51in2026Nice}
\min_{P \in \simplexabs{Q}{\set{Y}}} & &  \int \hat{h}_{\frac{\alpha_1}{\lambda_1},\frac{\alpha_2}{\lambda_2}}(x_1, x_2, y)  \mathrm{d}P(y) + \KL{P}{Q}.
\end{IEEEeqnarray}
%
%
This structural property sheds light on the open question concerning the choice of the reference measure $Q$ in \eqref{EqMarch5at14h56in2026SophiaA}, which is a central question in statistical learning theory. See for instance, \cite{perlaza2024empirical, bermudez2026machine, perlaza2023validation, bermudez2026decentralized}, and references therein.
\subsection{Nested Gibbs Probability Measures}

Nested Gibbs probability measures can be formally defined as follows.
\begin{definition}[Nested Gibbs Probability Measures]
\label{DefNestedGibbsMeasures}
An $(h, Q_1, \lambda)$-Gibbs probability measure $P^{(h, Q_1, \lambda)}_{Y \mid X = x}\in \triangle\left(\set{Y} \right)$, of the form  in~\eqref{EqGeneralGibbsRN}, is said to be a nested Gibbs probability measure if $Q_1$ is a Gibbs probability measure.
\end{definition}
In the following, the measure $P^{(h_1,Q_1,\lambda_1)}_{Y_1\mid X_1 = x_1}$ in \eqref{EqApril15at18h05in2026Sophia} is said to be nested within $P^{(h_2,Q_2,\lambda_2)}_{Y_2\mid X_2 = x_2}$ in~\eqref{EqApril22at21h10in2026Nice}, when $Q_2$ is chosen to be identical to $P^{(h_1,Q_1,\lambda_1)}_{Y_1\mid X_1 = x_1}$, which yields a $\autoparent{h_2,P^{(h_1,Q_1,\lambda_1)}_{Y_1\mid X_1 = x_1},\lambda_2}$-Gibbs probability measure, denoted by
\begin{IEEEeqnarray}{rCl}
\label{EqApril15at18h06in2026Sophia}
P^{\autoparent{h_2,P^{(h_1,Q_1,\lambda_1)}_{Y_1\mid X_1 = x_1},\lambda_2}}_{Y_2\mid X_2 = x_2}
\in
\simplex{\set{Y}}.
\end{IEEEeqnarray}
Such a measure is a nested Gibbs probability measure (Definition~\ref{DefNestedGibbsMeasures}).
The following theorem provides an explicit expression for the Radon-Nikodym derivative of $P^{\autoparent{h_2,P^{(h_1,Q_1,\lambda_1)}_{Y_1\mid X_1 = x_1},\lambda_2}}_{Y_2\mid X_2 = x_2}$  with respect to the measure~$Q_1$.
\begin{theorem}\label{TheoremApril16at15h22in2026Antibes}
The nested Gibbs probability measure $P^{\autoparent{h_2,P^{(h_1,Q_1,\lambda_1)}_{Y_1\mid X_1 = x_1},\lambda_2}}_{Y_2\mid X_2 = x_2}$ in~\eqref{EqApril15at18h06in2026Sophia} satisfies for all $y\in\supp Q_1$,
\begin{IEEEeqnarray}{rCl}
\label{EqNestedExplicit}
\nonumber
&&
\RND{P^{\autoparent{h_2,P^{(h_1,Q_1,\lambda_1)}_{Y_1\mid X_1 = x_1},\lambda_2}}_{Y_2\mid X_2 = x_2}}{Q_1}\autoparent{y}
\\
&=&
\frac{\exp\autoparent{-\frac{1}{\lambda_1} h_1\autoparent{x_1,y} - \frac{1}{\lambda_2} h_2\autoparent{x_2,y}}}
{\int \exp\autoparent{-\frac{1}{\lambda_1} h_1\autoparent{x_1,\tilde{y}} - \frac{1}{\lambda_2} h_2\autoparent{x_2,\tilde{y}}}\,\mathrm{d}Q_1(\tilde{y})}\\
\label{EqMarch18at9h48in2026Sophia}
& = & 
\RND{P_{Y\mid (X_1,X_2)=(x_1,x_2)}^{\autoparent{\hat{h}_{\frac{1}{\lambda_1},\frac{1}{\lambda_2}},Q_1,1}}}{Q_1} (y),
\end{IEEEeqnarray}
where $P_{Y\mid (X_1,X_2)}^{\autoparent{\hat{h}_{\frac{1}{\lambda_1},\frac{1}{\lambda_2}},Q_1,1}}$ is an $(\hat{h}_{\frac{1}{\lambda_1},\frac{1}{\lambda_2}},Q_1,1)$-Gibbs conditional probability measure in $\simplex{\set{Y} | \set{X} \times \set{X}}$; and the function $\hat{h}_{\frac{1}{\lambda_1},\frac{1}{\lambda_2}}$ is defined in \eqref{EqEffectivehApril14}.
\end{theorem}
\begin{IEEEproof} 
The proof is presented in \cite[Appendix~D]{InriaRR9613}.
\end{IEEEproof}
\enlargethispage{-0.10in}
Theorem~\ref{TheoremApril16at15h22in2026Antibes} establishes that the nested Gibbs probability measure $P_{Y_2\mid X_2 = x_2}^{\autoparent{h_2, P^{(h_1,Q_1,\lambda_1)}_{Y_1\mid X_1 = x_1}, \lambda_2}}$ is identical to an $(\hat{h}_{\frac{1}{\lambda_1}, \frac{1}{\lambda_2}},Q_1,1)$-Gibbs probability measure $P_{Y\mid (X_1,X_2)=(x_1,x_2)}^{(\hat{h}_{\frac{1}{\lambda_1}, \frac{1}{\lambda_2}},Q_1,1)}$. More formally, 
\vspace{-2ex}
\begin{IEEEeqnarray}{rCl}
\KL{P_{Y_2\mid X_2 = x_2}^{\autoparent{h_2, P^{(h_1,Q_1,\lambda_1)}_{Y_1\mid X_1 = x_1}, \lambda_2}}}{P_{Y\mid (X_1,X_2)=(x_1,x_2)}^{\autoparent{\hat{h}_{\frac{1}{\lambda_1}, \frac{1}{\lambda_2}},Q_1,1}}} =  0.
\end{IEEEeqnarray}
\enlargethispage{-0.02in}
\vspace{-2ex}

Hence, the effect of nesting $P^{(h_1,Q_1,\lambda_1)}_{Y_1\mid X_1 = x_1}$ in \eqref{EqApril15at18h05in2026Sophia} into the $(h_2,Q_2,\lambda_2)$-Gibbs probability measure $P_{Y_2\mid X_2 = x_2}^{\autoparent{h_2, Q_2, \lambda_2}}$ in~\eqref{EqApril22at21h10in2026Nice}, i.e., choosing $Q_2$ identical to $P^{(h_1,Q_1,\lambda_1)}_{Y_1\mid X_1 = x_1}$, has a twofold effect on $P_{Y_2\mid X_2 = x_2}^{\autoparent{h_2, Q_2, \lambda_2}}$. 
First, it transforms the objective function $h_2$ into a linear combination of $h_1$ and $h_2$, with coefficients $\frac{1}{\lambda_1}$ and $\frac{1}{\lambda_2}$, which yields $\hat{h}_{\frac{1}{\lambda_1}, \frac{1}{\lambda_2}}$ in \eqref{EqEffectivehApril14}.
Second, it changes the regularization factor from $\lambda_2$ to $1$.
These observations, together with \cite[Lemma~1]{PerlazaEntropy2025}, lead to the conclusion that the nested Gibbs probability measure~$P_{Y_2\mid X_2 = x_2}^{\autoparent{h_2, P^{(h_1,Q_1,\lambda_1)}_{Y_1\mid X_1 = x_1}, \lambda_2}}$ is the unique solution, if it exists, to one of the following problems:
\begin{IEEEeqnarray}{lCl}
\label{EqApril22at15h33in2026Nice}
\left\lbrace
\begin{array}{l}
\hspace{-1ex}\displaystyle\min_{P\in\simplexabs{Q_1}{\set{Y}}}  \int h_2\autoparent{x_2,y}\,\mathrm{d}P(y) +\lambda_2 \KL{P}{P^{(h_1,Q_1,\lambda_1)}_{Y_1\mid X_1 = x_1}},  \mbox{ $\lambda_2 > 0$} \Dsupersqueezeequ \\
\hspace{-1ex}\displaystyle\max_{P\in\simplexabs{Q_1}{\set{Y}}}  \int h_2\autoparent{x_2,y}\,\mathrm{d}P(y) +\lambda_2 \KL{P}{P^{(h_1,Q_1,\lambda_1)}_{Y_1\mid X_1 = x_1}},  \mbox{ $\lambda_2 < 0$}, \spnum \Dsupersqueezeequ
\end{array}
\right. \hspace{-4ex}
\end{IEEEeqnarray}
and at the same time, the unique solution, if it exists, to the minimization problem:
\begin{IEEEeqnarray}{rCl}
\label{EqMarch17at6h43in2026BusToSophiaA}
\min_{P \in \simplexabs{Q_1}{\set{Y}}} & &  \int \hat{h}_{\frac{1}{\lambda_1},\frac{1}{\lambda_2}}(x_1, x_2, y)  \mathrm{d}P(y) + \KL{P}{Q_1},
\end{IEEEeqnarray}
which is reminiscent of the optimization problem in \eqref{EqApril22at22h51in2026Nice}.
In particular, note that these problems become the same if $\alpha_1 = \alpha_2 = 1$ and $Q$ and $Q_1$ are identical. This observation leads to the following corollary, which establishes an equivalence between a  normalized log-linear combination of $P^{(h_1,Q_1,\lambda_1)}_{Y_1\mid X_1 = x_1}$ in \eqref{EqApril15at18h05in2026Sophia} and $P^{(h_2,Q_2,\lambda_2)}_{Y_2\mid X_2 = x_2}$ in~\eqref{EqApril22at21h10in2026Nice} and the nested Gibbs probability measure $P^{(h_2,P^{(h_1,Q_1,\lambda_1)}_{Y_1\mid X_1 = x_1},\lambda_2)}_{Y_2\mid X_2 = x_2}$ in~\eqref{EqApril15at18h06in2026Sophia}.
More specifically, denote by $S_2 \in \simplexabs{Q_1}{\set{Y}}$, the $(1,1,Q_1,Q_2,Q_1)$-normalized log-linear combination of $P^{(h_1,Q_1,\lambda_1)}_{Y_1\mid X_1 = x_1}$ and $P^{(h_2,Q_2,\lambda_2)}_{Y_2\mid X_2 = x_2}$. Hence, for all $y \in \supp Q_1$,
\vspace{-2ex}
\begin{IEEEeqnarray}{rCl}
\label{EqApril15at22h01in2026Antibes}
\RND{S_2}{Q_1}\autoparent{y}
&=&
\frac{
\exp\autoparent{
-\frac{1}{\lambda_1}h_1\autoparent{x_1,y}
-\frac{1}{\lambda_2}h_2\autoparent{x_2,y}}}
{\int
\exp\autoparent{
-\frac{1}{\lambda_1}h_1\autoparent{x_1,\hat{y}}
-\frac{1}{\lambda_2}h_2\autoparent{x_2,\hat{y}}}
\mathrm{d}Q_1\autoparent{\hat{y}}}  \spnum\\
\label{EqApril15at22h01in2026AntibesB}
& = & 
\RND{P^{\autoparent{h_2,P^{(h_1,Q_1,\lambda_1)}_{Y_1\mid X_1 = x_1},\lambda_2}}_{Y_2\mid X_2 = x_2}}{Q_1}\autoparent{y},
\end{IEEEeqnarray}
\enlargethispage{-0.02in}
\noindent
where the equality in~\eqref{EqApril15at22h01in2026Antibes} follows from Theorem~\ref{TheoremApril16at15h19in2026Antibes};
and the equality in~\eqref{EqApril15at22h01in2026AntibesB} follows from Theorem~\ref{TheoremApril16at15h22in2026Antibes}.
 Using this notation, the following holds.
\begin{corollary}
\label{CorollaryApril16at15h20in2026Antibes}
The nested Gibbs probability measure $P^{(h_2,P^{(h_1,Q_1,\lambda_1)}_{Y_1\mid X_1 = x_1},\lambda_2)}_{Y_2\mid X_2 = x_2}$ in~\eqref{EqApril15at18h06in2026Sophia} and the Gibbs probability measure $S_2$ in~\eqref{EqApril15at22h01in2026Antibes} are identical.  That is,
\begin{IEEEeqnarray}{rcl}
\KL{S_2}{P^{(h_2,P^{(h_1,Q_1,\lambda_1)}_{Y_1\mid X_1 = x_1},\lambda_2)}_{Y_2\mid X_2 = x_2}} & = & 0.
\end{IEEEeqnarray}

\end{corollary}

\section{Applications in One-Shot Federated Learning}
\enlargethispage{-0.10in}
Consider a federated learning system in which $K$ clients collaboratively tune
their local learning algorithms by communicating with a common server. Let 
$\set{M}\subseteq\reals^{d}$, with $d\in\ints$, and let $\set{X}$; and 
$\set{Y}$, denote the sets of \emph{models}, \emph{patterns}, and
\emph{labels}, respectively. For all $k \inCountK{K}$, client $k$ has $n_k$
training data points $(x_{k,1}, y_{k,1}),\ldots,(x_{k,n_k}, y_{k,n_k})$,
which are elements of the set $\set{Z}\triangleq\set{X}\times\set{Y}$, and form the
dataset
\vspace{-1ex}
\begin{IEEEeqnarray}{rCl}
\label{EqDatasetK}
\vect{z}_k & \triangleq & \autoparent{(x_{k,1}, y_{k,1}), \ldots, (x_{k,n_k}, y_{k,n_k})}\in \set{Z}^{n_k}.
\end{IEEEeqnarray}
The aggregation of all local training datasets is denoted by
\vspace{-1ex}
\begin{IEEEeqnarray}{rCl}
\label{EqDatasetZero}
\vect{z}_0 & \triangleq & \autoparent{\vect{z}_1, \ldots, \vect{z}_K } \in \set{Z}^{n_0},
\end{IEEEeqnarray}
with $n_0 \triangleq \sum_{k=1}^{K} n_k$.
Given a model $\vect{\theta}\in\set{M}$, the loss induced by such a model on a
data point~$(x,y)\in\set{Z}$ is denoted by $\ell(x,y,\vect{\theta})$, where $\ell:\set{Z}\times\set{M}\to[0,+\infty)$ is referred to as the \emph{loss function}, which is assumed to be Borel
measurable.

The \emph{empirical risk} induced by a model with respect to an $m$-length
dataset is determined by the function
\vspace{-1ex}
\begin{IEEEeqnarray}{rCl}
\label{EqEmpiricalRisk}
\mathsf{L}_{m}:\function{\set{Z}^{m}\times\set{M}}{[0,+\infty)}{\autoparent{\vect{z},\vect{\theta}}}
{\frac{1}{m}\sum_{i=1}^{m}\ell\autoparent{x_{i},y_{i},\vect{\theta}}.}
\end{IEEEeqnarray}

A supervised learning algorithm trained upon $n_k$-length
datasets is represented by a conditional probability measure in
$\simplex{\set{M} \mid \set{Z}^{n_k}}$.
A class of algorithms that is central in this section is that of Gibbs
algorithms. For all $k \inCountK{K}$, given a $\sigma$-finite measure $Q_k$
defined on $\set{M}$, a regularization factor
$\lambda_k \in (0,+\infty)$, and a fixed dataset
$\vect{z}_k \in \set{Z}^{n_k}$, the instance of the Gibbs algorithm at
client~$k$, trained upon the dataset $\vect{z}_k$, is represented by the
Gibbs probability measure\vspace{-1ex}
\begin{IEEEeqnarray}{rCl}
\label{EqGibbsClientKFixedDataset}
P^{(\mathsf{L}_{n_k},Q_k,\lambda_k)}_{\vect{\Theta}_k\mid \vect{Z}_k=\vect{z}_k}
&\in&
\simplexabs{Q_k}{\set{M}}.
\end{IEEEeqnarray}
%
Consider a one-shot federated learning system (two communication phases). In the first phase, clients transmit their local Gibbs algorithms to the server. During the second phase, the server aggregates these local algorithms by means of a normalized log-linear combination and broadcasts the resulting probability measure to all clients.
More specifically, given a $\sigma$-finite measure $Q$ defined on $\mathcal{M}$, such that $Q\ll Q_k$ for all $k \inCountK{K}$ and the coefficients $\alpha_1,\ldots,\alpha_K\in\mathbb {R}\setminus \{0\}$, the server defines the measure
$S_{K}\in\simplexabs{Q}{\set{M}}$ such that for all $\vect{\theta} \in \supp Q$,
\vspace{-2ex}
\begin{IEEEeqnarray}{rCl}
\label{EqApril23at18h56in2026Antibes}
\RND{S_{K}}{Q}(\vect{\theta})
&=&
\frac{
\displaystyle\prod_{k=1}^{K}
\left(\RND{P^{\autoparent{\mathsf{L}_{n_k}, Q_k, \lambda_k}}_{\vect{\Theta}_k\mid \vect{Z}_k=\vect{z}_k}}{Q_k}(\vect{\theta})
\right)^{\alpha_k}}
{\displaystyle\int\prod_{k=1}^{K}\left(\RND{P^{\autoparent{\mathsf{L}_{n_k}, Q_k, \lambda_k}}_{\vect{\Theta}_k\mid \vect{Z}_k=\vect{z}_k}}{Q_k}(\vect{\nu}) \right)^{\alpha_k} dQ(\vect{\nu})}
\end{IEEEeqnarray}
\begin{IEEEeqnarray}{rCl}
\label{EqApril23at18h31in2026Antibes}
&=&
\frac{
\exp\left(
-\displaystyle\sum_{k=1}^{K}\frac{\alpha_k}{ \lambda_k}
\mathsf{L}_{n_k}(\vect{z}_k,\vect{\theta})
\right)}
{\displaystyle\int
\exp\left(
-\sum_{k=1}^{K}\frac{\alpha_k}{\lambda_k}
\mathsf{L}_{n_k}(\vect{z}_k,\vect{\nu})
\right)
\,\mathrm{d}Q(\vect{\nu})},
\end{IEEEeqnarray}
where, for all $k \inCountK{K}$, it has been assumed that $Q \ll Q_k$. The equality
in~\eqref{EqApril23at18h31in2026Antibes} follows from iteratively using Theorem~\ref{TheoremApril16at15h19in2026Antibes}.
Moreover, from \cite[Lemma~1]{PerlazaEntropy2025}, it follows that the measure $S_{K}$ is the unique solution to
\begin{IEEEeqnarray}{rCl}
\label{EqApril26at20h06in2026Antibes}
\min_{P\in\simplexabs{Q}{\set{M}}}
&&
\int
\left(
\sum_{k=1}^{K}\frac{\alpha_k}{ \lambda_k}
\mathsf{L}_{n_k}(\vect{z}_k,\vect{\theta})
\right)\mathrm{d}P(\vect{\theta})
+
\KL{P}{Q}.\spnum
\end{IEEEeqnarray}

The following corollary formalizes this observation, which represents a centralized-performance guarantee of the federated learning system depicted above.

\begin{corollary}
\label{CorApril23at18h55in2026Antibes}
Assume that for all $k\inCountK{K}$, 
\begin{IEEEeqnarray}{rcL}
\lambda_k = \frac{1}{n_k}; \mbox{ and } \alpha_k = \frac{1}{n_0\lambda_0},
\end{IEEEeqnarray}
 in \eqref{EqApril23at18h56in2026Antibes}, for some $\lambda_0>0$. Then, for all $\vect{\theta}\in \supp Q$,
\begin{IEEEeqnarray}{rCl}
\label{EqApril23at18h47in2026Antibes}
\RND{S_{K}}{Q}(\vect{\theta})
&=&
\frac{
\exp\left(
\frac{-1}{\lambda_0}~\mathsf{L}_{n_0}(\vect{z}_0,\vect{\theta})
\right)}
{\displaystyle\int
\exp\left(
\frac{-1}{\lambda_0}~\mathsf{L}_{n_0}(\vect{z}_0,\vect{\nu})
\right)
\,\mathrm{d}Q(\vect{\nu})} = 
\RND{
P^{(\mathsf{L}_{n_0},Q,\lambda_0)}_{\vect{\Theta}\mid \vect{Z}=\vect{z}_0}
}{Q}(\vect{\theta}). \middlesqueezeequ \spnum
\end{IEEEeqnarray}
\end{corollary}

Corollary~\ref{CorApril23at18h55in2026Antibes} shows that
 the normalized log-linear combination of the $K$ algorithms independently obtained by the clients by training Gibbs algorithms upon their local training datasets, namely 
the measure $S_K$ in \eqref{EqApril23at18h31in2026Antibes}, is identical to a Gibbs algorithm trained on the aggregated dataset $\vect{z}_0$.
\enlargethispage{-0.10in}
More generally, the main assumption in Corollary~\ref{CorApril23at18h55in2026Antibes} is that, for all $k \inCountK{K}$,
\vspace{-3ex}
\begin{IEEEeqnarray}{rCl}
\frac{\alpha_k}{\lambda_k} = \frac{n_k}{n_0 \lambda_0}.
\end{IEEEeqnarray}
This condition can be satisfied through several choices of the pairs $(\lambda_k,\alpha_k)$. Interestingly, $\lambda_k$ is chosen locally by client $k$, whereas $\alpha_k$ is chosen by the server. In particular, the choice $\lambda_k = 1/n_k$ depends only on the size of the local dataset at client $k$. Under this choice, the condition reduces to
$\alpha_k = \frac{1}{n_0\lambda_0}$,
which is independent of $k$. Hence, the server does not need any client-specific information about the local datasets; it only requires the total number of samples $n_0$ and the centralized regularization parameter $\lambda_0$.
In other words, in a one-shot federated learning setting, this choice of parameters allows the server to achieve the same performance as a centralized Gibbs algorithm trained on the aggregation of all training datasets, by forming a normalized log-linear combination of the locally trained Gibbs algorithms.
In particular, this occurs by sharing probability distributions on the set of models only and not the actual data.
The feasibility of such a learning system is constrained by the possibility of transmitting the corresponding Gibbs measures without distortion, which is not possible in practice. This limitation is related to the practical issue mentioned in~\cite{bermudez2026decentralized}. The effect of such a distortion remains to be formally studied.

\clearpage
\IEEEtriggeratref{15}
\bibliographystyle{\Latexfilepath/Bibliography/IEEEtranlink}
\bibliography{\Latexfilepath/Bibliography/Merged_BERMUDEZ_ref.bib}

\begin{thebibliography}{10}
\providecommand{\url}[1]{#1}
\csname url@samestyle\endcsname
\providecommand{\newblock}{\relax}
\providecommand{\bibinfo}[2]{#2}
\providecommand{\BIBentrySTDinterwordspacing}{\spaceskip=0pt\relax}
\providecommand{\BIBentryALTinterwordstretchfactor}{4}
\providecommand{\BIBentryALTinterwordspacing}{\spaceskip=\fontdimen2\font plus
\BIBentryALTinterwordstretchfactor\fontdimen3\font minus
  \fontdimen4\font\relax}
\providecommand{\BIBforeignlanguage}[2]{{%
\expandafter\ifx\csname l@#1\endcsname\relax
\typeout{** WARNING: IEEEtran.bst: No hyphenation pattern has been}%
\typeout{** loaded for the language `#1'. Using the pattern for}%
\typeout{** the default language instead.}%
\else
\language=\csname l@#1\endcsname
\fi
#2}}
\providecommand{\BIBdecl}{\relax}
\BIBdecl

\bibitem{xu2017information}
A.~Xu and M.~Raginsky, ``Information-theoretic analysis of generalization
  capability of learning algorithms,'' in \emph{Proceedings of the
  International Conference on Neural Information Processing Systems (NeurIPS)},
  vol.~30, Long Beach, CA, USA, Dec. 2017, pp. 2521--2530.

\bibitem{russo2016controlling}
D.~Russo and J.~Zou, ``Controlling bias in adaptive data analysis using
  information theory,'' in \emph{Proceedings of the 19th International
  Conference on Artificial Intelligence and Statistics}, vol.~51, Cadiz, Spain,
  May 2016, pp. 1232--1240.

\bibitem{medina2022robustness}
M.~Avella~Medina, J.~L. Montiel~Olea, C.~Rush, and A.~Velez, ``On the
  robustness to misspecification of {$\alpha$}-posteriors and their variational
  approximations,'' \emph{Journal of Machine Learning Research}, vol.~23, no.
  147, pp. 1--51, 2022.

\bibitem{Bu2024Towards}
Y.~Bu, ``Towards optimal inverse temperature in the {G}ibbs algorithm,'' in
  \emph{Proceedings of the IEEE International Symposium on Information Theory
  (ISIT)}, Athens, Greece, Jul. 2024, pp. 2257--2262.

\bibitem{bu2023generalization}
Y.~Bu, H.~V. Tetali, G.~Aminian, M.~Rodrigues, and G.~W. Wornell, ``On the
  generalization error of meta learning for the {G}ibbs algorithm,'' in
  \emph{Proceedings of the IEEE International Symposium on Information Theory
  (ISIT)}, Taipei, Taiwan, Jun. 2023, pp. 2488--2493.

\bibitem{perlaza2023validation}
S.~M. Perlaza, I.~Esnaola, G.~Bisson, and H.~V. Poor, ``On the validation of
  {G}ibbs algorithms: Training datasets, test datasets and their aggregation,''
  in \emph{Proceedings of the IEEE International Symposium on Information
  Theory (ISIT)}, Taipei, Taiwan, Jun. 2023, pp. 328--333.

\bibitem{catoni2007pac}
O.~Catoni, \emph{{PAC-Bayesian} Supervised Classification: The Thermodynamics
  of Statistical Learning}, 1st~ed., ser. IMS Lecture Notes--Monograph
  Series.\hskip 1em plus 0.5em minus 0.4em\relax Beachwood, OH, USA: Institute
  of Mathematical Statistics, 2007, vol.~56.

\bibitem{ray2023asymptotics}
R.~Ray, M.~A. Medina, and C.~Rush, ``Asymptotics for power posterior mean
  estimation,'' in \emph{Proceedings of the 59th Annual Allerton Conference on
  Communication, Control, and Computing (Allerton)}, Monticello, IL, USA, Sep.
  2023, pp. 1--8.

\bibitem{perlaza2024empirical}
S.~M. Perlaza, G.~Bisson, I.~Esnaola, A.~Jean-Marie, and S.~Rini, ``Empirical
  risk minimization with relative entropy regularization,'' \emph{IEEE
  Transactions on Information Theory}, vol.~70, no.~7, pp. 5122--5161, Jul.
  2024.

\bibitem{rodriguezgalvez2024information}
B.~Rodr{\'i}guez-G{\'a}lvez, ``An information-theoretic approach to
  generalization theory,'' Ph.D. dissertation, KTH Royal Institute of
  Technology, 2024.

\bibitem{Azizian2024What}
W.~Azizian, F.~Lutzeler, J.~Malick, and P.~Mertikopoulos, ``What is the
  long-run distribution of stochastic gradient descent? {A} large deviations
  analysis,'' in \emph{Proceedings of the International Conference on Machine
  Learning (ICML)}, Vienna, Austria, Jul. 2024, pp. 2168--2229.

\bibitem{aminian2021exact}
G.~Aminian, Y.~Bu, L.~Toni, M.~Rodrigues, and G.~Wornell, ``An exact
  characterization of the generalization error for the {G}ibbs algorithm,'' in
  \emph{Proceedings of the International Conference on Neural Information
  Processing Systems (NeurIPS)}, vol.~34, Virtual Event, Dec. 2021, pp.
  8106--8118.

\bibitem{aminian2024information}
G.~Aminian, Y.~Bu, L.~Toni, M.~R.~D. Rodrigues, and G.~W. Wornell,
  ``Information-theoretic characterizations of generalization error for the
  {G}ibbs algorithm,'' \emph{IEEE Transactions on Information Theory}, vol.~70,
  no.~1, pp. 632--655, Jan. 2024.

\bibitem{perlaza2024generalization}
S.~M. Perlaza and X.~Zou, ``The method of gaps: Exact expressions for the
  generalization error of supervised learning algorithms,'' \emph{To appear in
  IEEE Transactions on Information Theory}, 2026.

\bibitem{zou2024Generalization}
X.~Zou, S.~M. Perlaza, I.~Esnaola, and E.~Altman, ``Generalization analysis of
  machine learning algorithms via the worst-case data-generating probability
  measure,'' in \emph{Proceedings of the AAAI Conference on Artificial
  Intelligence}, vol.~38, no.~15, Vancouver, Canada, Feb. 2024, pp.
  17\,271--17\,279.

\bibitem{zouJSAIT2024}
X.~Zou, S.~M. Perlaza, I.~Esnaola, E.~Altman, and H.~V. Poor, ``The worst-case
  data-generating probability measure in statistical learning,'' \emph{IEEE
  Journal on Selected Areas in Information Theory}, vol.~5, pp. 175--189, Apr.
  2024.

\bibitem{bercher2012simple}
J.-F. Bercher, ``A simple probabilistic construction yielding generalized
  entropies and divergences, escort distributions and q-{G}aussians,''
  \emph{Physica A: Statistical Mechanics and its Applications}, vol. 391,
  no.~19, pp. 4460--4469, Oct. 2012.

\bibitem{wilson1983renormalization}
K.~G. Wilson, ``The renormalization group and critical phenomena,''
  \emph{Reviews of Modern Physics}, vol.~55, no.~3, pp. 583--600, Jul. 1983.

\bibitem{asadi2025hierarchical}
A.~R. Asadi, ``Hierarchical maximum entropy via the renormalization group,''
  \emph{arXiv preprint arXiv:2509.01424}, Sep. 2025.

\bibitem{chhabra1989direct}
A.~Chhabra and R.~V. Jensen, ``Direct determination of the f($\alpha$)
  singularity spectrum,'' \emph{Physical Review Letters}, vol.~62, no.~12, pp.
  1327--1330, Mar. 1989.

\bibitem{beck1993thermodynamics}
C.~Beck and F.~Schl{\"o}gl, \emph{Thermodynamics of Chaotic Systems: An
  Introduction}, ser. Cambridge Nonlinear Science Series.\hskip 1em plus 0.5em
  minus 0.4em\relax Cambridge, UK: Cambridge University Press, 1993, vol.~4.

\bibitem{abe2003geometry}
S.~Abe, ``Geometry of escort distributions,'' \emph{Physical Review E},
  vol.~68, no.~3, p. 031101, Sep. 2003.

\bibitem{ohara2010dually}
A.~Ohara, H.~Matsuzoe, and S.~Amari, ``A dually flat structure on the space of
  escort distributions,'' \emph{Journal of Physics: Conference Series}, vol.
  201, no.~1, p. 012012, Dec. 2010.

\bibitem{bercher2009source}
J.-F. Bercher, ``Source coding with escort distributions and {R}{\'e}nyi
  entropy bounds,'' \emph{Physics Letters A}, vol. 373, no.~36, pp. 3235--3238,
  Aug. 2009.

\bibitem{tsallis2009introduction}
C.~Tsallis, \emph{Introduction to Nonextensive Statistical Mechanics:
  Approaching a Complex World}, 1st~ed.\hskip 1em plus 0.5em minus 0.4em\relax
  New York, NY, USA: Springer, 2009.

\bibitem{abe2005necessity}
S.~Abe and G.~B. Bagci, ``Necessity of q-expectation value in nonextensive
  statistical mechanics,'' \emph{Physical Review E}, vol.~71, no.~1, p. 016139,
  Jan. 2005.

\bibitem{bermudez2026machine}
Y.~Bermudez, S.~M. Perlaza, and I.~Esnaola, ``Machine unlearning for {G}ibbs
  supervised learning algorithms,'' in \emph{Proceedings of the IEEE
  International Symposium on Information Theory (ISIT)}, Guangzhou, China, Jun.
  2026.

\bibitem{bermudez2026decentralized}
------, ``Decentralized machine learning with centralized performance
  guarantees via {G}ibbs algorithms,'' in \emph{Proceedings of the IEEE
  International Symposium on Information Theory (ISIT)}, Guangzhou, China, Jun.
  2026.

\bibitem{lalitha2018social}
A.~Lalitha, T.~Javidi, and A.~D. Sarwate, ``Social learning and distributed
  hypothesis testing,'' \emph{IEEE Transactions on Information Theory},
  vol.~64, no.~9, pp. 6161--6179, Sep. 2018.

\bibitem{mcmahan2017communication}
H.~B. McMahan, E.~Moore, D.~Ramage, S.~Hampson, and B.~{Ag{\"u}era y Arcas},
  ``Communication-efficient learning of deep networks from decentralized
  data,'' in \emph{Proceedings of the 20th International Conference on
  Artificial Intelligence and Statistics (AISTATS)}, ser. Proceedings of
  Machine Learning Research, vol.~54, Apr. 2017, pp. 1273--1282.

\bibitem{InriaRR9613}
Y.~Bermudez, S.~M. Perlaza, and I.~Esnaola, ``Equivalence between nested
  {G}ibbs measures and log-linear combinations of {G}ibbs measures,'' INRIA,
  Centre Inria d'Université Côte d'Azur, Sophia Antipolis, France, Tech. Rep.
  RR-9613, May 2026.

\bibitem{PerlazaEntropy2025}
S.~M. Perlaza and G.~Bisson, ``Variations on the expectation due to changes in
  the probability measure,'' \emph{Entropy}, vol.~27, no. 8:865, pp. 1--20,
  Aug. 2025.

\end{thebibliography}
\end{document}